\documentclass[conference]{IEEEtran}

\IEEEoverridecommandlockouts
 \pdfoutput=1
\usepackage{graphicx}

\graphicspath{{./figures/}}
\usepackage{caption}
\usepackage{subcaption}
\usepackage{float}
\usepackage{amsmath} % assumes amsmath package installed
\usepackage{amssymb}  % assumes amsmath package installed
\usepackage{amsthm}
\usepackage{amsfonts}
\usepackage{cite}
\usepackage{url}
\usepackage{algorithm}
\usepackage[noend]{algpseudocode}
\usepackage{color}
\usepackage{array}
\usepackage{fancyhdr}
\newtheorem{mytheom}{Theorem}
\newtheorem{lemma}{Lemma}
\newtheorem{proposition}{Proposition}

\newtheorem{defn}{Definition}

\newtheorem{assumption}{Assumption}
\title{Enhanced Automatic Generation Control (E-AGC) for  Electric Power Systems with Large Intermittent Renewable Energy Sources
\thanks{The authors greatly appreciate partial funding by NIST.  Also, discussions with  Ms. Rupamathi Jaddivada and the use of  Scalable Electric Power System Simulator (SEPSS) at MIT  \cite{SEPSS} are greatly acknowledged.  
}
	\thanks{Accepted to the IEEE PES GM 2019. \textcopyright  ~2019 IEEE. }
}

	\author{	\IEEEauthorblockN{ Xia Miao, Marija Ili\'{c}}
	\IEEEauthorblockA{\textit{LIDS, Massachusetts Institute of Technology}\\
		 \{xmiao,ilic\}@mit.edu}
		\and
		\IEEEauthorblockN{ Qixing Liu}
			\IEEEauthorblockA{\textit{China Southern Power Grid}, Guangzhou, China \\
				liuqx@csg.cn}
   }
\begin{document}
	\maketitle

\begin{abstract}
This paper is motivated by the need to enhance today's Automatic Generation Control (AGC) for ensuring high quality frequency response  in the  changing  electric power systems.  Renewable energy sources, if not controlled carefully, create persistent fast and often large oscillations in  their electric power outputs.  A sufficiently detailed dynamical model  of the interconnected system which captures  effects of  fast nonlinear disturbances created by the renewable energy resources  is  derived for the first time. Consequently,  the real power flow interarea oscillations, and the resulting frequency deviations are modeled. The modeling is multi-layered, and the dynamics of each layer (component level (generator); control area (control balancing authority), and the interconnected system) is  expressed in terms of  internal states and the interaction variables (IntV) between the layers and within the layers. E-AGC is then derived using this model to show how these interarea oscillations can be canceled.  Simulation studies are carried out on a 5-bus system.
\end{abstract}

\section{Introduction and motivation}
\label{Sec:Intro}
A high quality of electricity service requires near-ideal nominal frequency, which is achieved by maintaining instantaneous supply-demand power balance. System operation under off-nominal frequency can deteriorate electric equipment, degrade the performance of electric load and even lead to wide-spread system failures and blackouts \cite{blackout}. Recently, the industrial concerns regarding  frequency quality have grown as the increasing Renewable Energy Sources (RES) presence. The RES which are inherently  intermittent can lead to continuous supply-demand mismatch and drive the system frequency varying around the desired nominal value with unacceptable quality of response (QoR). 

To secure power system operations, the unacceptable frequency excursion must be regulated close to zero in real time by means of automated feedback control. The AGC is widely implemented for this purpose \cite{AGC1,AGC2}. However, AGC is mainly designed based on steady state concepts. When AGC is applied to a system with RES, the fast persistent disturbances caused by the RES can drive the system dynamically varying around the equilibrium such that the assumptions of the AGC could become invalid and the AGC might not be as effective as expected. Therefore, the frequency regulation needs to be enhanced and the new approach should extend the modeling and control of AGC from steady state to dynamics.

In the past decades, to improve the performance of AGC, a concept of Area Control Error (ACE) Diversity Interchange (ADI) was proposed in the industry practice \cite{ADI}. However, it is still based on steady state concepts and it is not economic efficient as no coordination exists between different area. A LQR-based full state feedback control was proposed in \cite{LQR} for load frequency control. Thereafter, many follow-up works have been done. One limitation of LQR-based approaches is that they are centralized and requiring overly complicated sensing and communication. To improve QoR and system level coordination without complicated sensing and communication infrastructure, the authors introduce Enhanced AGC (E-AGC) concept in \cite{ Qixing2012, QixingThesis}. However, it should be pointed out that almost all of these design utilize the linearized model which is not valid for large disturbances. Thus, a new approach is needed to consider the tradeoff between the control performance and the complexity.

There are two main contributions of this paper. First, a sufficiently detailed nonlinear dynamical model which captures effects of fast and large nonlinear disturbances is derived for the first time in Section \ref{Sec:model}. Notably, most of existing frequency regulation methods ignore the network dynamics. In Section \ref{Sec:simulation}, we show that the fast network dynamics should not be neglected in today's electric power systems due to high RESs penetration. Second, we adopt the merits of \cite{Qixing2012,QixingThesis} and then propose a general multi-layered control using the proposed nonlinear interconnected system model in Section \ref{Sec:ControlDesign}. We have also provably shown that the proposed approach is capable of canceling fast network inter-area dynamical oscillations and achieving system-level coordination. Unlike \cite{Qixing2012},  the proposed method no longer utilizes the small signal model and the routinely made assumptions such as the network dynamics are non-oscillatory. Thus, the contributions of this paper differ substantially from \cite{Qixing2012,QixingThesis}.  

The paper is organized as follows. Section \ref{Sec: problemformulation} and \ref{Sec:model} provide problem formulation and dynamic model. The proposed multi-layered control is explained in Section \ref{Sec:ControlDesign}. Case studies are given in Section \ref{Sec:simulation}. Section \ref{Sec:Conclusion} concludes the paper.  
\section{Dynamic Modeling and Problem Formulation }
\label{Sec: problemformulation}
%\subsection{Modeling assumption}
Power electronic devices have been widely installed in the field for voltage regulation. Thus, it is reasonable to make: 
\begin{assumption}
	The voltage magnitude of each bus is bounded. 
\end{assumption}

It should be noted that Assumption 1 is a mild assumption as we do not fix the voltage. It can vary within a range. In fact, most of existing frequency control approaches neglect the voltage dynamics, i.e., voltage is constant. Notice that the network dynamics is considered. As  the interconnected system is treated as a dynamical system, each bus is equivalent. In other words, traditional bus classification (PV, PQ bus) is no longer suitable.  

\subsection{Dynamical model of system components}
\subsubsection{Generation component}
Generators have similar role contributing to frequency dynamics, regardless of their types. Thus, we choose a nonlinear non-reheat generator with a primary governor controller embedded as \cite{QixingThesis}:
\begin{equation}
\begin{split}
\label{Eqn:SM}
\dot \delta_G &= \omega_0(\omega_G-\omega^{ref})\\\
M\dot \omega_G &= P_m + P_m^{ref} - D(\omega_G - \omega_0) - P_e \\
T_u \dot{P}_m &= -P_m + K_t a \\
T_g \dot{a} &= -ra - (\omega_G -\omega^{ref}) + u_{AGC}	
\end{split}
\end{equation}  
State variables $x_G = [\delta_G, \omega_G, P_m, a]^T$ represent the rotor angle, rotational speed, mechanical power injection, and steam valve position, respectively. $\omega_0$ is the rated angular velocity. $M$, $D$, $K_t$, $T_u$, $T_g$ and $r$ are machine parameters. 

%$M$ and $D$ denote the inertia constant and damping coefficient, respectively. $T_u$ is the time constant reflecting the response time to the valve position change of the turbine nozzles. $K_t$ is a proportionality factor representing the control valve position variation relative to the turbine output variation. $T_g$ is the time constant of the valve servomotor turbine gate system. $r$ is the permanent speed droop of the turbine. $\omega^{ref}$ is the set-point of the local speed governor. $u_{AGC}$ represents the potential AGC control signal. 

It should be pointed out that $P_e$ is the source of the nonlinearity for \eqref{Eqn:SM}. $P_e = f_1(\delta_G, x_{TL}, x_{L})$ represents the sum of the real power transfered on its connecting transmission lines, which is a nonlinear function of  $\delta_G$, line states $x_{TL}$ and load states $x_{L}$. 

\subsubsection{Load component}
The load is modeled in the network reference frame as:  
\begin{equation}
\begin{split}
\label{Eqn:load}
L_{L} \dot i_{Ld} = -R_{L} i_{Ld} + \omega L_{L}i_{Lq} + V_{Ld}  \\
L_{L} \dot i_{Lq} = -R_{L} i_{Lq} - \omega L_{L}i_{Ld} + V_{Lq}   
\end{split}
\end{equation}
State variables $x_L = [i_{Ld}, i_{Lq}]^T$ represent the d-axis and q-axis load current, respectively. $L_{L}$ and $R_L$ stand for the load inductance and resistance. $\omega$ denotes the grid frequency. $V_{Ld}$ and $V_{Lq}$ are the d-axis and q-axis of the terminal voltage. $V_{Ld}:=V\cos\theta_V$ and $V_{Lq} := V\sin\theta_V$ are nonlinear function of the terminal voltage angle $\theta_V$. Note that $\theta_V =\delta_G$ when a generator is connected at the same bus. The load \eqref{Eqn:load} satisfies: 

\begin{proposition}
	\label{prop:1}
	Given Assumption 1, state variables $x_{L}$ of load component \eqref{Eqn:load} are bounded.  
\end{proposition}

\begin{proof}
	It can be seen that \eqref{Eqn:load} is asymptotically stable if $V_{Ld} = 0$ and $V_{Lq} = 0$. In addition, system matrix $A_L$ is Hurwitz. Thus, using Corollary 5.2 in \cite{Khaili}, we know that \eqref{Eqn:load} is $\mathcal{L}_p$ stable. Notice that voltage magnitude  $V = \sqrt{V_{Ld}^2+V_{Lq}^2}$. Provided Assumption 1, nonlinear inputs $V_{Ld}$ and $V_{Lq}$ are bounded, which yields that $x_{L}$ are bounded. 
\end{proof}

\subsubsection{Network component (transmission line)}
 Transmission line component is modeled in the network reference frame as: 
\begin{equation}
\label{Eqn:TL}
\begin{split}
	L_{TL} \dot i_{TLd} = -R_{TL} i_{TLd} + \omega L_{TL}i_{TLq} + V_{d,L} - V_{d,R} \\
		L_{TL} \dot i_{TLq} = -R_{TL} i_{TLq} - \omega L_{TL}i_{TLd} + V_{q,L} - V_{q,R}    
\end{split}
	\end{equation}
State variables $x_{TL} = [i_{TL,d}, i_{TL,q}]^T$ represent the d and q-axis line current. $R_{TL}$ and $L_{TL}$ are resistance and inductance of the line. $(V_{d,L}, V_{q,L})$ and $ (V_{d,R}, V_{q,R})$ denote the left and right port voltage, respectively. The nonlinearity of \eqref{Eqn:TL} is introduced by its port voltages. 

\begin{proposition}
	\label{prop:2}
	Given Assumption 1, state variables $x_{TL}$ of network dynamics \eqref{Eqn:TL} are bounded.  
\end{proposition}
	The proof is similar as that of Proposition \ref{prop:1}.  Detail derivation is omitted for brevity.

\subsection{Modeling of disturbances}
Disturbances are characterized as exogenous hard-to-predict inputs to the system. Since disturbances can enter the system through different components, we group them into a vector of external disturbances  $d_{ext}$ as seen by components. 
\subsection{Dynamical model of interconnected systems}
The overall interconnected system dynamics can be obtained by combining components together as:
\begin{equation}
\label{Eqn:system}
\begin{split}
	\dot{x}_G & = A_Gx_G + B_Gu_{AGC} + F_Gf_1(x_G,x_{TL},x_{L}, d_{ext})\\
	\dot{x}_{TL} &= A_{TL}x_{TL} + F_{TL}f_2(x_{TL},x_G,x_{L}, d_{ext}) \\
	\dot{x}_{L} &= A_Lx_L + F_Lf_3(x_{L},x_{TL},x_G, d_{ext})
\end{split}
\end{equation}

Notably, $A_G$ is rank 1 deficiency due to the conservation of power, while $A_{TL}$ and $A_{L}$ are Hurwitz matrices. $F_G$, $F_{TL}$ and $F_L$ are the input matrices corresponding to nonlinear coupling $f_1$, $f_2$ and $f_3$, respectively. Network coupling between different components are implicitly shown in $f_1$, $f_2$ and $f_3$. 

\subsection{Problem formulation}
The problem considered in this paper can be posed as: 
\begin{itemize}
	\item Given:  interconnected system  dynamical model \eqref{Eqn:system}
	\item Design: AGC control input $u_{AGC}$
	\item Objectives: both state variables $[x_G, x_{TL},x_L]^T$ and nonlinear interaction $[f_1, f_2, f_3]$ are stabilized and regulated. 
\end{itemize}

\section{Multi-layered dynamical model of interconnected systems}
\label{Sec:model}
In this section the multi-layered model of the interconnected systems \eqref{Eqn:system} is derived. Notice that variations of $f_2$ and $f_3$ are indeed driven by $f_1$, due to the fact that generators are the only active components that produce power. In addition, we have shown that $x_{TL}$ and $x_{L}$ are bounded by $f_2$ and $f_3$ (see Proposition 1 and 2). Therefore, if  $f_1$ can be controlled, $f_2$ and $f_3$ can be indirectly controlled, which further ensures the system performance.  

To achieve this goal, the definition of the interaction variable (IntV), which was proposed in \cite{IntV2,QixingThesis},  is revisited and a new interpretation is proposed for the nonlinear systems \eqref{Eqn:system}.
\begin{defn} Given a dynamic component (subsystem), its IntV $z$ is an output variable in terms of the local states of the component (subsystem) and it satisfies:
	\begin{equation}
		z \equiv const
	\end{equation}
	when the component (subsystem) is free of any conserved net power imbalance.
\end{defn}
An IntV is generally defined to capture the non-zero conserved net power imbalance of a component (subsystem). In what follows, the multi-layered dynamic model based on IntV is provided. 

We first decompose $u_{AGC}$ of \eqref{Eqn:SM} into component-level, area-level, and system-level control signal as: 
\begin{equation}
\label{Eqn:uAGC}
u_{AGC} = u_{AGC, c} + u_{AGC,r}  + u_{AGC,s}
\end{equation}
These control components will later appear at different layers. 
\subsection{Component-level dynamical model}
 Component-level IntV dynamical model has the form: 
\begin{equation}
\label{Eqn:componentIntV}
\dot{z}_c = P_m^{ref} -P_e + \frac{K_t}{r}u_{AGC,c}  \quad z_c(t_0) =z_{c0}
\end{equation}
 It can been seen that $\dot{z}_c$ captures the conserved net power imbalance of the component. It is worthwhile mentioning that $z_c$ simply depends on its own states, i.e., no assumption about the strength of the external interconnection is needed. This fact has been proved for linearized models \cite{IntV2,QixingThesis}.  
\subsection{Area-level dynamical model}
Similarly, a new IntV $z_r$ is introduced for the control areas. Recall the steady-state concept ACE.  The dynamics of $z_r$ can be therefore considered as a dynamic version of ACE. 

The dynamics of IntV $z_r^C$ is: 
\begin{equation}
\label{Eqn:areaInt}
\dot{z}_r^C = \sum_{i=1}^{N_c^r}\dot z_{c,i} = B_r \mathbf{u}_{AGC,r}  ~~~  z_r(t_0) =z_{r0},~B_r = \mathbf{1}^{N_c^r\times 1}
\end{equation}
$N_c^r$ stands for the number of generators inside the control area. 
\subsection{System-level dynamical model}
We can then apply the same procedure at the interconnected system level. Thus, the dynamic of  system-level IntV $z_s$ is: 
\begin{equation}
\label{Eqn:sysInt}
\dot{z}_s^R= \sum_{i=1}^{N_s^R}\dot z_{r,i}^C  = B_s \mathbf{u}_{AGC,s}~~~  z_s(t_0) =z_{s0},~B_s = \mathbf{1}^{N_s^r\times 1}
\end{equation}
where $N_s^R$ is the number of control areas.

\section{Design of Enhanced AGC(E-AGC) for complex electric power system dynamics}
\label{Sec:ControlDesign}
Primarily, the objective of the E-AGC is to ensure an acceptable QoR of frequency dynamics. This is achieved through controlling the IntVs at different levels. First, $z_c$ is controlled at constant in order to eliminate the real-time net power imbalance. Second, $z_r^C$ and $z_s^R$ need to be regulated to zero in order to maintain the variation of total inadvertent power exchange around zero. Through the coordination of  widely dispersed control resources, the inexpensive ones can be fully utilized so that the system-level control cost can be reduced in comparison to today's AGC approach.
\subsection{Component-level design}
The component level IntV dynamics \eqref{Eqn:componentIntV} is utilized. In order to have constant $z_c$, we design $u_{AGC,c}$ as:
\begin{equation}
\label{Eqn:u_c}
	u_{AGC,c} = \frac{r}{K_t}(P_e - P_m^{ref})
\end{equation} 

Substituting Eqn.\eqref{Eqn:u_c} into Eqn.\eqref{Eqn:SM}, we obtain the closed-loop generator model, which is provably stable under certain conditions. The result is given below.   

\begin{lemma}
	\label{lemma1}
	With control design Eqn.\eqref{Eqn:u_c}, the generator module \eqref{Eqn:SM} is stable in the sense of Lyapunov if the following condition is satisfied: 
	\begin{equation}
		||P_e - P_m^{ref}||_2 \leq \frac{K_t}{r}u_{max}
	\end{equation}
	where $u_{max}$ denotes the saturation limit of the control input. 
\end{lemma}

Proof of Lemma \ref{lemma1}  is organized in the appendix.  

\subsection{Area-level coordination}
The objective of this layer is to eliminate the conserved net power imbalance of each area by optimally controlling $z_r^C$.  Within each control area, in order to obtain the optimal coordinated law among participating generators, we design the following LQR problem:
\begin{eqnarray}
\label{Eqn:u_r}
	\min_{u_{AGC,r}}   &J = \int_{t_0}^{\infty}[(z_r^C)^TQ_rz_r^C  + (u_{AGC,r})^TR_ru_{AGC,r}]~d\tau \notag\\
	s.t. \quad & \dot{z}_r^C = u_{AGC,r} \quad z_r^C(t_0) = z_{r0} 
\end{eqnarray}
$R_r$ specifies the weight of control cost of each generator. In practice, these two matrices are tunable under the constraint that $Q_r$ and $R_r$ are positive definite matrices. 
\subsection{System-level coordination}
The objective of this layer is to eliminate the conserved net power imbalance of the overall system by optimally controlling $z_s^R$. The control areas are coordinated through exchanging their IntVs and controlling the IntVs that are collected.  Similarly, we apply the LQR technique to optimize the following objective function: 
\begin{eqnarray}
\label{Eqn:u_s}
\min_{u_{AGC,s}}   &J = \int_{t_0}^{\infty}[(z_s^R)^TQ_sz_s^R  + (u_{AGC,s})^TR_su_{AGC,s}] ~d\tau \notag\\
s.t. \quad & \dot{z}_s^R = u_{AGC,s} \quad z_s^R(t_0) = z_{s0}
\end{eqnarray}
$R_s$ defines the relative control cost between different control areas. In operation, $Q_s$ and $R_s$ can be tuned accordingly. 

\subsection{Main theoretical result of the  E-AGC}
In this section, we give the main theoretical result of the proposed E-AGC approach. 
\begin{mytheom}
	Given Assumption 1 and the composite control design \eqref{Eqn:u_c} - \eqref{Eqn:u_s}, the interconnected dynamical system \eqref{Eqn:system} will be stabilized and the frequency of each generator will be regulated if the following condition is satisfied: 
		\begin{equation}
		||P_e - P_m^{ref}||_2 \leq \frac{K_t}{r}u_{max}
		\end{equation}
\end{mytheom}
Given the page limit, we provide a sketch for the proof:
\begin{proof}
	Notice that that $z_s$ and $z_r$ can be provably regulated via LQR problems. Thus, as $t \to \infty$, $z_r \to 0$ and $z_s \to 0$.  Recall the IntV definition and Lemma 1. It can be concluded that $z_c\to 0$ which indicates $\omega \to \omega^{ref}$.  
\end{proof}

\subsection{Communication and implementation discussion}
The communication infrastructures shown in Fig.\ref{fig:5bus} enable the measurement and also the exchange of IntVs for implementing the E-AGC.
\begin{figure}[htp]
	\centering
	\includegraphics[width=0.35\textwidth]{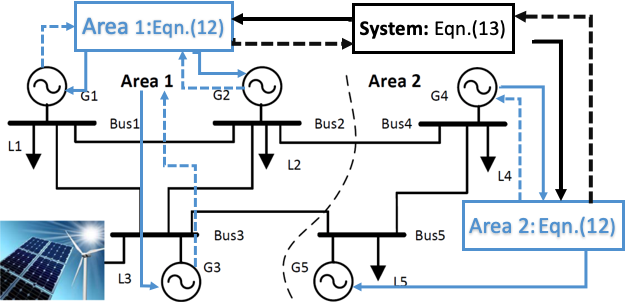}
	\caption{Information exchange of the E-AGC on a 5 bus system }
	\label{fig:5bus}
\end{figure}

 First, each generator measure its local state variables and then use \eqref{Eqn:componentIntV} to obtain the IntV. Once an IntV is locally computed, a synchronized time-stamp should be added, and then sent to its control area (doted blue line). The control area has to compute its IntV and then compute the coordinated control signals using \eqref{Eqn:areaInt} and \eqref{Eqn:u_r}.  Similarly, after receiving the IntV from control areas (doted black line), the central coordinator computes the system-level coordinated control signals using \eqref{Eqn:u_s} and then distributed back to the control areas (solid black line). Each control area further provides each generator with a control signal comprised of the control signals from different levels (solid blue line). It should be emphasized that only the IntVs are exchanged between different layers. Thus, we minimize the required information exchange, which is also safe from the cyber security perspective. 

Note that the proposed control  does not change existing governor control. Different layer signals will be added using \eqref{Eqn:uAGC} and then be applied to \eqref{Eqn:SM} as a composite control. Governor set points $\omega^{ref}$ are fixed. In addition, the proposed approach does not require predefined area power set points, unlike some hierarchical control approaches where the secondary layer needs tertiary level power set points.  Therefore, the proposed control indeed achieves both stabilization and regulation.

\section{Illustration of the E-AGC on a 5-Bus System }
\label{Sec:simulation}
In this section, simulation studies on a 5-bus (two-area) test system (Fig.\ref{fig:5bus}) are carried out. The nonlinear system \eqref{Eqn:system} including network dynamics is simulated using SEPSS at MIT\cite{SEPSS}. The purpose are twofold: to show the importance of network dynamics in frequency regulation and to illustrate the effectiveness of the proposed E-AGC. 
\subsection{System description and the test scenario}

The total capacity of the system is 25 MW with 20\% of the electric energy provided by  the RES installed at bus 3. As shown in Fig.\ref{fig:5bus}, two areas are interconnected via two transmission lines.  We assume that two control areas are strongly connected, while components are weakly connected within the area. 

In order to show the effect of network dynamics, we only consider the step changes at all loads. For this scenario, the conventional AGC is supposed to restore the frequency.  

In what follows, the simulation results of the cases with no AGC, the conventional AGC, and the proposed E-AGC are shown. As the low frequency oscillations are observed in operation and our simulations, we then provide an explanation of why they have not been captured by classic methods. 
\subsection{Simulation results and discussion}

\subsubsection{Performance with the conventional AGC}
\begin{figure}[htp!]
	\centering %left bottom right top
	\begin{subfigure}[b]{0.22\textwidth}
		\includegraphics[width=\textwidth]{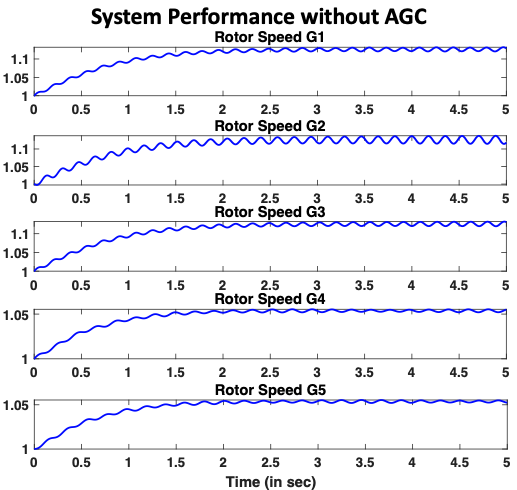}
		\caption{With primary control only}
		\label{fig:NoAGC}
	\end{subfigure}
	~
	\begin{subfigure}[b]{0.2\textwidth}
		\includegraphics[width=\textwidth]{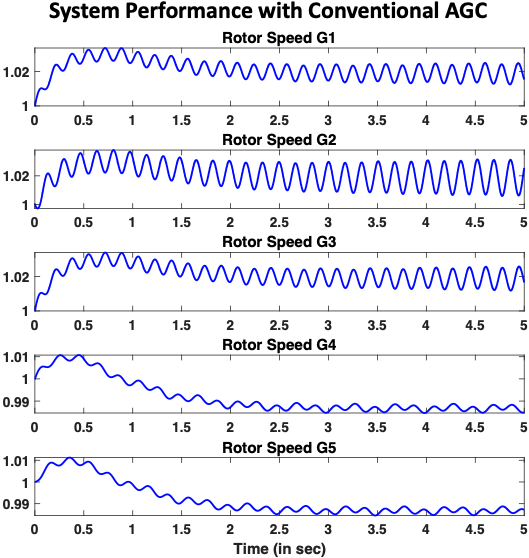}
		\caption{With conventional AGC}
		\label{fig:AGC}
	\end{subfigure}
	\caption{Frequency responses of the 5 bus system}
\end{figure}

We first disable the AGC and simulate the system with primary controllers only. Frequency responses are organized in Fig.\ref{fig:NoAGC}. It can be seen that the frequency of the generators in Area 1 settles around 1.1 $p.u.$ but has $2-5~Hz$ oscillations. Similarly, the frequency of Area 2 is oscillating around 1.05 $p.u$.

Next, we activate the conventional AGC \cite{Kundur}.  Corresponding frequency responses are shown in Fig.\ref{fig:AGC}. It can be seen that steady-state errors are greatly reduced. However, the frequency of Area 1 is higher than the nominal value, while Area 2 is slightly lower. This indicates that the inter-area oscillation exists between two areas. It is because the RES provides more power than what Area 1 needs. 

It should be also noted that the low frequency oscillations observed in Fig.\ref{fig:NoAGC} still exist in Fig.\ref{fig:AGC}.  It is worthwhile mentioning that such low frequnecy oscillations never show up in classic analysis but system operators do observe similar phenomena in operation. This is because most of conventional approaches are designed based on the quasi-static ACE and the network dynamics is ignored. However, we only assume that voltage magnitude is bounded in our model. In other words, d-q axis voltage can vary over time. As shown in the line dynamics \eqref{Eqn:TL}, the varying voltage angle may act as disturbances  to the component. Recall Proposition 1 and 2. They both explain why we observe oscillatory but bounded behavior in simulations. Therefore, it is important to consider network dynamics into frequency analysis. otherwise such unobserved oscillations are likely large enough to trigger protection devices. 

Zooming into control design, we notice that neither the primary control nor the conventional AGC has feedback with respect to rotor angle, i.e., rotor angle (voltage angle) is not directly controlled. Hence, voltage angle may interact with $x_{TL}$ and then start to oscillate. In other words, real and reactive power produced at one bus are interacting with the energy stored in the line when the oscillation occurs. Consequently, disturbances at one bus may spread out to other buses through line dynamics, which further cause oscillatory behavior in the entire system. Fig.\ref{fig:AGC} also supports the fact that there is no guarantee that output feedback can stabilize the rotor angle.  

 One argument for  ignoring network dynamics is that the network has much smaller time constant compared to generators. However, transmission line dynamics cannot change instantaneously in reality and the argument is no longer true in microgrids. The rate of real and reactive power entered from two ends of the line are not necessary to be the same. Thus, fast disturbances introduced by RESs may excite  the fast network dynamics,  resulting in accumulated effects on the slow dynamics (frequency dynamics). This will become a critical issue if more and more RESs are integrated.
 
It should be mentioned that the performance can potentially be improved if rotor angle deviation is considered in the feedback design. It is equivalent to design a PI controller. However, there are several challenges in implementing this solution. First of all, it is hard to get accurate rotor angle reference. It may not be realistic to run centralized optimization (such as AC OPF) after every change in the system. Second, the feedback gain with respect to rotor angle needs to be carefully design. The gain should be tunned so that it can tolerate large and fast-varying disturbances. Improperly tunned integrator can destabilize the system. Last but not the least, it is challenging to measure rotor angle accurately without large delay.  

\subsubsection{Performance with the proposed E-AGC}
Simulation results of the proposed E-AGC are given in Fig.\ref{fig:EAGC}. 
\begin{figure}[htp!]
	\centering
	\includegraphics[width=0.22\textwidth]{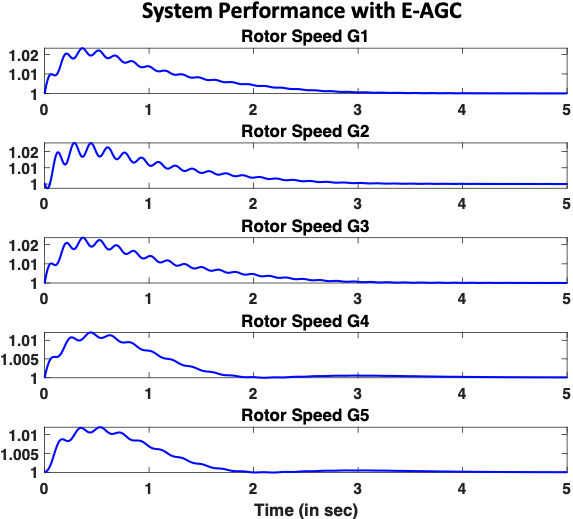}
	\caption{Frequency responses with E-AGC}
	\label{fig:EAGC}
\end{figure}

In comparison to the conventional AGC, frequency of each generators are regulated to the nominal value. More importantly, low frequency oscillations no longer exist. This indicates that the imbalances within and between control areas are limited around zero, i.e., no inter-area oscillations.  It is because the proposed E-AGC is designed based on nonlinear dynamic systems. Network dynamics is preserved in the dynamics of the IntV at different levels. At regulation stage, instead of requiring hard-to-get angle reference, area-level and system-level control are using the IntV as feedback signals.  These two layers not only coordinate the resources, but also act as an integrator, which eventually eliminates the low frequency oscillations observed in Fig.\ref{fig:AGC}. 

From economic point of view, the proposed E-AGC significantly reduce the systematic regulation cost via area-level and system-level coordination, due to the proposed  formulation. Considering that the E-AGC requires much less information exchange, we argue that this proposed control scheme could be very cost-effective. In addition, we believe that the IntV information is one potential communication protocol for the future grid operation, grid control, etc. 

\section{Conclusions}
\label{Sec:Conclusion}
In this paper we revisit  the frequency regulation problem for future electric energy systems. We summarize  the emerging practical problems of applying the conventional AGC, especially when network dynamics and highly variable RESs are presented in the system. The E-AGC approach is thus introduced as an alternative solution. %The proposed approach is designed based using the notion of IntV whose dynamics that  captures the interactions evolving at different levels of the hierarchical electric power grid. 
The regulation cost can be systematically reduced by using little information exchange. Simulations show that the E-AGC outperforms the conventional approaches. %, as the E-AGC fully eliminates low frequency oscillations, inter-area dynamical oscillations and steady state errors. 
Simulations for large-scale systems will be given in our future publications.   

%\section*{Acknowledgments}
%
%The authors greatly appreciate partial funding by the US National Institute of Standards and Technology.  Also, discussions with Ms. Rupamathi Jaddivada and the use of  SEPSS at MIT  \cite{SEPSS}) are greatly acknowledged.  

\appendix
\subsection{Proof Sketch for Lemma 1}
\begin{proof}
 $T_1 = 
[\frac{D+Kt/r}{M}~M~T_u~ \frac{T_gK_t}{r}
]$, $T_2 = [
0~M~T_u~\frac{T_gK_t}{r}]^T$ and $P = (T_{2}T_{1})^T(T_2T_1)$. It is easy to check that $P\in \mathbb{R}_{+}^{4\times4}$. Thus, if we choose a Lyapunov function $V = x_G^TPx_G$, 
the rest is straightforward to show using the procedures in  \cite{Khaili}. 
\end{proof}

\end{document}